\documentclass[a4paper,11pt]{article}

\usepackage{amsmath,amsfonts, amssymb}
\usepackage{mathtools}
\usepackage[pdftex]{graphicx}
\usepackage{color}
\usepackage[english]{babel}
\usepackage{tabularx}
\usepackage[table]{xcolor}
\usepackage[left=2.5cm,top=4cm,right=2.5cm,bottom=4cm]{geometry}  
\usepackage[hidelinks]{hyperref}
\usepackage{array}
\usepackage{sectsty}
\usepackage{amsthm}
\usepackage{enumerate}
\usepackage{braket}
\usepackage{tikz}
\usetikzlibrary{positioning}
\usetikzlibrary{arrows.meta}
\usepackage[style = tikz]{mdframed}

\usepackage{indentfirst}
\usepackage[nottoc]{tocbibind}

\usepackage[title]{appendix}
\usepackage[font=small,labelfont=bf,tableposition=top]{caption}

\usepackage{mathtools}
\DeclarePairedDelimiter{\ceil}{\lceil}{\rceil}

\usepackage{todonotes}

\makeatletter
\newtheorem*{rep@theorem}{\rep@title}
\newcommand{\newreptheorem}[2]{%
\newenvironment{rep#1}[1]{%
 \def\rep@title{#2 \ref{##1}}%
 \begin{rep@theorem}}%
 {\end{rep@theorem}}}
\makeatother

\newtheorem{theorem}{Theorem}
\newreptheorem{theorem}{Theorem}

\newtheorem{lem}[theorem]{Lemma}
\newtheorem{cor}[theorem]{Corollary}

\newreptheorem{observation}{Observation}
\newreptheorem{corollary}{Corollary}

\theoremstyle{definition}
\newtheorem{definition}[theorem]{Definition}

\theoremstyle{remark}

\newtheorem{claim}{Claim}

\newenvironment{subproof}[1][\proofname]{%
  \begin{proof}[#1]%
}{%
  \end{proof}%
}

\DeclareMathOperator{\2}{\{0,1\}}

\DeclareMathOperator{\Hilb}{\mathcal{H}}

\DeclareMathOperator{\sol}{sol}
\DeclareMathOperator{\poly}{poly}
\DeclareMathOperator{\Ptime}{\textsc{P}}
\DeclareMathOperator{\NP}{\textsc{NP}}
\DeclareMathOperator{\BQP}{\textsc{BQP}}
\DeclareMathOperator{\BPP}{\textsc{BPP}}
\DeclareMathOperator{\FBQP}{\textsc{FBQP}}
\DeclareMathOperator{\FBPP}{\textsc{FBPP}}
\DeclareMathOperator{\FP}{\textsc{FP}}
\DeclareMathOperator{\FNP}{\textsc{FNP}}
\DeclareMathOperator{\QMA}{\textsc{QMA}}
\DeclareMathOperator{\PP}{\textsc{PP}}
\DeclareMathOperator{\QCMA}{\textsc{QCMA}}
\DeclareMathOperator{\PH}{\textsc{PH}}
\DeclareMathOperator{\SAT}{SAT}
\DeclareMathOperator{\len}{len}

\DeclareMathOperator{\overlap}{\textsc{overlap}}
\DeclareMathOperator{\sat}{\textsc{sat}}
\DeclareMathOperator{\search}{\textsc{search}}

\title{$\BQP$, meet $\NP$: Search-to-decision reductions and approximate counting}
\author{Sevag Gharibian \footnote{Department of Computer Science and Institute for Photonic Quantum Systems (PhoQS), Paderborn University,
Germany. Email: \{sevag.gharibian, jonas.kamminga\}@upb.de.} \and Jonas Kamminga $^*$}

\date{}

\begin{document}
\maketitle
\begin{abstract}
    What is the power of polynomial-time quantum computation with access to an NP oracle?
    In this work, we focus on two fundamental tasks from the study of Boolean satisfiability (SAT) problems: search-to-decision reductions, and approximate counting.
    We first show that, in strong contrast to the classical setting where a poly-time Turing machine requires $\Theta(n)$ queries to an NP oracle to compute a witness to a given SAT formula, quantumly $\Theta(\log n)$ queries suffice.
    We then show this is tight in the black-box model --- any quantum algorithm with ``$\NP$-like'' query access to a formula requires $\Omega(\log n)$ queries to extract a solution with constant probability.

    Moving to approximate counting of SAT solutions, by exploiting a \emph{quantum} link between search-to-decision reductions and approximate counting, we show that existing classical approximate counting algorithms are likely optimal. First, we give a lower bound in the ``$\NP$-like'' black-box query setting: Approximate counting requires $\Omega(\log n)$ queries, even on a quantum computer. We then give a ``white-box'' lower bound (i.e. where the input formula is not hidden in the oracle) --- if there exists a randomized poly-time classical or quantum algorithm for approximate counting making $o(\log n)$ $\NP$ queries, then $\BPP^{\NP[o(n)]}$ contains a $\Ptime^{\NP}$-complete problem if the algorithm is classical and $\FBQP^{\NP[o(n)]}$ contains an $\FP^{\NP}$-complete problem if the algorithm is quantum.
\end{abstract}

\section{Introduction}\label{scn:intro}

A fundamental direction of study in classical complexity theory is: What can P or BPP achieve with access to an NP oracle? Here, the study of relational-problems, i.e. where the output is not a single bit, but a string, has proven particularly fruitful. (Formally, this refers to the classes FunctionP ($\FP$) and FunctionBPP ($\FBPP$); see Section \ref{sec:prelim}.) A first direction here has been \emph{search-to-decision} reductions. Namely, given a SAT formula $\varphi$ and an NP oracle, it is well-known that $O(n)$ queries suffice to extract a solution to $\varphi$ (assuming one exists). Moreover, this is likely classically optimal: an $o(n)$-query algorithm would violate the Exponential Time Hypothesis~\cite{impagliazzo2001complexity}\footnote{If there is an $\FP^{\NP[o(n)]}$ algorithm outputting a satisfying assignment then SAT can be decided in time $2^{o(n)}$ as follows: Enumerate through all possible strings $y$ of $o(n)$ NP query answers, which takes $2^{o(n)}$ time. For each $y$, run the $\FP$ machine on $y$ to obtain candidate solution $x$, and check if $\varphi(x)=1$.}. (Before ETH was posited,  Krentel showed that if $O(\log n)$ queries suffice, then $\Ptime = \NP$~\cite{krentel1986complexity}.)

A second key direction has been \emph{approximate counting} of the number of solutions of a Boolean formula, first studied by Stockmeyer~\cite{stockmeyer1983complexity}. Approximate counting has proven widely influential, even having applications in quantum advantage frameworks such as Aaronson and Arkhipov's Boson Sampling~\cite{aaronson2011computational}. Stockmeyer showed~\cite{stockmeyer1983complexity} that an $\FBPP$ machine making $O(\log n \log\log n)$ NP queries suffices to approximate the number of solutions within a constant multiplicative factor, and that at least $\Omega(\log n)$ queries are required. This gap was closed by Chakraborty, Meel, and Vardi, who improved the upper bound to $O(\log n)$ queries~\cite{chakraborty2016algorithmic}. Thus, the NP-query complexity of these two tasks is now well understood.\\
\vspace{-1mm}

\noindent \emph{The quantum setting.} The guiding question of this work is the next natural frontier: \emph{Can \emph{quantum} access to an NP oracle reduce the number of required queries?} For relation problems, this is a particularly intriguing question: Intuitively, a single classical NP query yields only $1$ bit of information, suggesting that if an $\FP$ machine wishes to produce an $n$-bit output, then $\Theta(n)$ queries are necessary. (Indeed, as mentioned above, this is the case for search-to-decision reduction of SAT, assuming ETH.) \emph{Quantum} access to an oracle, however, can sometimes bypass this obstacle, producing $n$ bit outputs with just a \emph{single} query. A notable instance of this is the Bernstein-Vazirani algorithm~\cite{bernstein1993quantum}, which requires just a single query to an oracle encoding an affine function $f(x)=a\cdot x + b$ to output string $x\in\set{0,1}^n$. We thus ask: \emph{Can a FunctionBQP ($\FBQP$) machine make fewer queries to an NP oracle to extract a SAT solution or approximately count the number of solutions?}

\paragraph{Our results.} In this work, we give tight resolutions to this question for both tasks. \\
\vspace{-1mm}

\noindent \emph{Search-to-decision reductions.} As mentioned earlier, classically, $\Theta(n)$ $\NP$ queries are necessary and sufficient for search-to-decision reduction for SAT, assuming ETH. Before proceeding to our main results, the lower bounds, we show that $O(\log n)$ queries suffice quantumly.
\begin{theorem}\label{thm:upperbound}
    $\FNP \subseteq \FBQP^{\NP[\log]}$.
\end{theorem}
\noindent Here, FunctionNP ($\FNP$) asks to produce a witness to an NP relation (Section \ref{sec:prelim}), and $\NP[\log]$ in the exponent denotes $O(\log n)$ $\NP$ queries. We remark that independently and prior to this work, Irani, Natarajan, Nirkhe, Rao and Yuen~\cite{irani2021quantum} showed that for SAT formulae with a \emph{unique} satisfying assignment, a single $\NP$ query suffices to extract said solution (see Related Work). 

Is Theorem~\ref{thm:upperbound} tight? Unfortunately, since we are in the \emph{white-box model} for search-to-decision reductions (i.e. the input formula $\varphi$ is given as input to the $\FBQP$ machine, rather then hidden in the oracle), even a single-query lower bound would imply\footnote{If one could show that any $\FBQP$ machine requires at least $1$ $\NP$ query for search-to-decision, then $\NP\not\subseteq\BQP$. This is because if $\NP\subseteq \BQP$, then $\FNP\subseteq\FBQP$ via the standard search-to-decision reduction for SAT.} $\NP \not\subseteq \BQP$, and is thus likely out of reach. We hence move to the \emph{black-box model} in order to prove a lower bound. For this, note that the standard quantum query model does not capture the power of ``existential'' or $\NP$ queries. Rather, we introduce the (straightforward quantum reformalization of) Stockmeyer's~\cite{stockmeyer1983complexity} existential query model:

\begin{definition}[Existential query model]
    An algorithm in the existential query model has access to the input string $x \in \2^N$ via the following existential query gate:
    \begin{equation}
        O^{\exists}_x \colon \ket{z} \mapsto (-1)^{\overlap(x,z)} \ket{z}
    \end{equation}
    where $z \in \2^N$ and $\overlap(x,z) = 1$ if there is an $i$ such that $x_i = z_i = 1$ and $0$ otherwise.
\end{definition}
\noindent 
In this model, we show a matching lower bound for Theorem~\ref{thm:upperbound}.
\begin{theorem}
    \label{blackboxlowerbound}
    Any quantum algorithm with existential query access to $x\in \2^N$ that outputs a $i$ with $x_i=1$ with constant probability needs to make $\Omega(\log\log N) = \Omega(\log n)$ existential queries.
\end{theorem}

\noindent\emph{Approximate counting.} Recall that, classically, approximate counting requires $\Theta(\log n)$ NP queries. We next exploit the fact that the technique behind the proof of Theorem \ref{thm:upperbound} (c.f.~\cite{irani2021quantum}) reveals a genuinely \emph{quantum} link between search-to-decision reduction and approximate counting. This allows us to show the following tight black-box lower bound on \emph{quantum} algorithms in the existential query model:
\begin{cor}
    \label{cor:blackboxnoapxcount}
    Any quantum algorithm with existential query access to a string $x \in \2^N$ which outputs an estimate $c$ such that $2^{|x| - 1} \le c < 2^{|x|}$, where $|x|$ is the Hamming weight of $x$, requires at least $\Omega(\log\log N)$ queries to the oracle.
\end{cor}
\noindent Above, $x$ denotes the truth table of the SAT formula, and so $N=2^n$ for $n$ the number of variables. Thus, for approximate counting, quantum algorithms do not outperform classical algorithms.

Finally, we prove a tight \emph{white-box} lower bound for \emph{both} classical or quantum algorithms which approximately count using $o(\log n)$ $\NP$ queries. As far as we are aware, no white-box lower bounds existed for either setting prior to this work.
\begin{cor}
    \label{cor:noapxcount}
    If there exists a \emph{classical} randomized poly-time algorithm for approximate counting, making $o(\log n)$ $\NP$ queries, then $\BPP^{\NP[o(n)]}$ contains a $\Ptime^{\NP}$-complete problem \footnote{Note that this is not quite as strong as $\Ptime^{\NP} \subseteq \BPP^{\NP[o(n)]}$ as overheads in the reduction to the $\Ptime^{\NP}$-complete problem may erase the reduction in the number of queries.}. Similarly, if there is a poly-time \emph{quantum} algorithm for approximate counting making $o(\log n)$ $\NP$ queries, then $\FBQP^{\NP[o(n)]}$ contains an $\FP^{\NP}$-complete problem.
\end{cor}
\noindent While the complexity theoretic implications above are not as standard as $\Ptime=\NP$ or the collapse of PH, they nevertheless would arguably be striking if true. This is because an $\FP^{\NP}$-complete problem is finding a satisfying assignment of \emph{smallest lexicographical ordering}~\cite{krentel1986complexity}. Thus, using $o(n)$ queries would seem to require resolving the lex-ordering in sublogarithmic time (in the search space size), whereas classical and quantum algorithms for the closely related task of binary search cannot achieve sublogarithmic time~\cite{ambainis1999better}.

\paragraph{Proof techniques.} We now sketch our proof techniques, organized by topic.\\
\vspace{-1mm}

\noindent\emph{Search-to-decision.} Theorem \ref{thm:upperbound} follows rather straightforwardly from prior results. We first note that, quantumly, the solution of a formula with a unique satisfying assignment can be found with a single $\NP$ query using the Bernstein Vazirani algorithm (c.f. \cite{irani2021quantum}). Therefore, it remains to reduce an arbitrary formula to a uniquely satisfiable one. Valiant and Vazirani showed this can be done with probability $O\left(\frac{1}{n}\right)$ \cite{valiant1985np}. If, however, the approximate cardinality of the set of solutions is known, then this reduction succeeds with constant probability. Since approximately counting this cardinality can be done with $\log(n)$ $\NP$ queries~\cite{stockmeyer1983complexity,chakraborty2016algorithmic} classically, this gives a quantum algorithm for search-to-decision reduction using $O(\log n)$ queries and with \emph{constant} success probability. The success probability can now be boosted to any constant by running the algorithm a constant number of times, checking the outputs and outputting one of the satisfying assignments.

Next, we discuss the first of our main results, the black-box lower bound (Theorem \ref{blackboxlowerbound}). Here the proof requires more work. Most quantum query lower bounds fall in one of two groups: polynomial methods and adversary methods \cite{beals2001quantum,bennett1997strengths,ambainis2000quantum}. Unfortunately, these methods are tailored to the standard query model, and it is not clear how to effectively utilize them in our \emph{existential} query model. Another complicating factor for us is that search-to-decision is a \emph{relational} problem, not a function problem. That is, for a given input formula $\varphi$ there are \emph{multiple} correct outputs: all solutions of $\varphi$.

To overcome this, we instead give a reduction from ``(unstructured) search with existential queries'' to ``binary search with standard queries'', so we may invoke Ambainis' binary search lower bound \cite{ambainis1999better}. We show that an algorithm for search on strings of length $N = 2^n$ using $q$ existential queries induces an algorithm for binary search on a space of size $n$ with the same number of \emph{standard} queries. The basic idea for this is as follows. If we can find a solution, then we can also sample a random solution by randomly permuting the solution space. Furthermore, a binary search instance, which is the task of finding the index of the first $1$ in a monotonically increasing binary string $x \in \2^n$, can be modified in the following way. We make a new exponentially longer string $y \in \2^{N}$ where the first $2^{n-1}$ entries of $y$ are set to $x_1$, the next $2^{n-2}$ to $x_2$ and so on. The index of a uniformly random $1$ in $y$ corresponds to the index of the least $1$ in $x$ with probability $>\frac{1}{2}$. Therefore, transforming $x$ into $y$ and running the random solution sampling algorithm on it solves binary search on $x$ using $q$ existential queries. We now note that because $x$ is monotonically increasing, any existential query can be simulated by a standard query. The results of an existential query with string $z$ will be the same as simply querying the largest index $i$ where $z_i = 1$. As the last step of the proof, we invoke Ambainis result that binary search on a space of size $n$ takes $\Omega(\log n)$ queries to complete the proof \cite{ambainis1999better}.\\
\vspace{-1mm}

\noindent\emph{Approximate counting.} Our black-box lower bound (Corollary \ref{cor:blackboxnoapxcount}) follows by combining the proof of Theorem \ref{thm:upperbound} with Theorem \ref{blackboxlowerbound}. If there is an approximate counter making $q$ existential queries, then an index with $x_i = 1$ can be found with constant success probability and $q+1$ existential queries using the algorithm from Theorem \ref{thm:upperbound}. By Theorem \ref{blackboxlowerbound} this is only possible if $q = \Omega(\log n)$.

Finally, we discuss our white-box lower bound (Corollary \ref{cor:noapxcount}). We assume the existence of an approximate counter making $o(\log n)$ queries and show that we can, with $o(\log n)$ queries, find the lexicographically smallest solution of a formula $\varphi$, which is an $\FP^{\NP}$-complete problem \cite{krentel1986complexity}. The main idea is as follows. The algorithm from Theorem \ref{thm:upperbound} samples approximately from the uniform distribution on the set of solutions of $\varphi$. We run this algorithm on the AND of $n^2$ copies of $\varphi$, where we pick new sets of variables for each instance. This will give us $n^2$ solutions of $\varphi$ picked almost uniformly at random. We find the least $x_{min}$ and repeat the process on the formula $\varphi(x) \wedge (x \le x_{min})$. After every round the number of solutions will be divided by at least $n$ with high probability. Therefore, after $O(\log_n(|\sol(\varphi)|)) = O\left(\frac{n}{\log n}\right)$ rounds we will have found the lexicographically smallest solution of $\varphi$. As every round takes $o(\log n)$ queries, we have found the lexicographically smallest solution using $o(n)$ queries, completing the proof.

\paragraph{Related work.} As previously mentioned, Irani, Natarajan, Nirkhe, Rao and Yuen already independently showed that the Bernstein-Vazirani algorithm can be used to find the solution of a \emph{uniquely} satisfiable formula (our Lemma \ref{iranilemma}) \cite{irani2021quantum}. They combine this with the Valiant-Vazirani theorem to do search-to-decision reduction for $\QCMA$ (and $\NP$) with a single query and \emph{inverse polynomial} success probability. However, they do not further study the case of a \emph{constant} success probability as done here. They also show that there exists a quantum polynomial time algorithm that makes a single query to a $\PP$ oracle and generates a witness for a $\QMA$ problem up to polynomial precision. Additionally, they show that there is an oracle such that $\QMA$ search does not reduce to $\QMA$ decision relative to that oracle. 

Search-to-decision reduction has been studied in other settings. If only parallel (i.e. non-adaptive) queries to the $\NP$ oracle are allowed, then the standard $O(n)$-query search-to-decision reduction for $\NP$ does not work. Nevertheless, it has been shown that $O(n^2)$ parallel oracle queries suffice for classical randomized algorithms \cite{ben1989theory}. Kawachi, Rossman and Watanabe showed that this is optimal in a black-box model and give an algorithm with improved error tolerance \cite{kawachi2012query}. In a later work they also consider more general black-box models and show that $O(n^2)$ parallel classical queries are still needed \cite{kawachi2017query}.

The class $\BQP$ with access to various resources has been studied before. Aaronson, Ingram and Kretschmer~\cite{aaronson2021acrobatics} study oracle separations between various complexity classes involving $\BQP$ as an oracle or $\BQP$ with access to an oracle. Among other results the authors prove that there is an oracle relative to which $\BQP^{\NP} \not\subseteq \PH^{\BQP}$ and an oracle relative to which $\NP^{\BQP} \not\subseteq \BQP^{\PH}$. Aaronson, Buhrman and Kretschmer~\cite{aaronson2023qubit} investigate $\BQP$ when given various types of advice. There it is shown, among other results, that $\textsc{FBQP/qpoly} \neq \textsc{FBQP/poly}$ (not relative to an oracle!).

Isolation algorithms, i.e., algorithms reducing the number of solutions of a Boolean formula to 1, have been studied by Dell, Kabanets, van Melkebeek and Watanabe~\cite{dell2013valiant}. They show that, unless $\NP \subseteq \textsc{P/poly}$, no randomized polynomial time isolation algorithm with success probability better than $\frac{2}{3}$ can exist.

\paragraph{Discussion and open questions.} Our paper characterizes the quantum NP-query complexity of search-to-decision reductions and approximate counting (and additionally gives a white-box lower bound for classical approximate counting algorithms). For this, some of our results utilized a quantum reformulation of Stockmeyer's classical existential query model. Can quantum query lower bound methods like polynomial methods and adversary methods be adapted to apply directly to existential queries? An obstacle here is the fact that, for example, adversary methods often keep track of how a ``progress measure'' increases with each query made. Existential queries, however, seem to allow arbitrarily large jumps in such ``progress measures''.

Second, we prove that $\BPP^{\NP[o(n)]}$ and $\FBQP^{\NP[o(n)]}$ containing $\Ptime^{\NP}$- and $\FP^{\NP}$-complete problems, respectively, are consequences to very efficient approximate counting. It would be interesting to see if there are further consequences to these conclusions or if our results can be strengthened. For example, can our results be strengthened to a contradiction of a common complexity theoretic hypothesis such as the (strong) Exponential Time Hypothesis?
Third, what other tasks might a $\BQP^{\NP}$ or $\FBQP^{\NP}$ machine be good for?
Finally, we close with a simple-to-state, concrete open question, which captures much of the difficulty of working with $\BQP^{\NP}$: Let $\varphi$ be a SAT formula. Classically, it is easy to see that a solution to $\varphi$ cannot be produced by an $\FBPP$ machine with a single $\NP$ query, for this would imply $\BPP=\NP$. This is because one can simply plug each possible answer, $0$ or $1$, from the $\NP$ machine into the $\FBPP$ machine, and check if the string $x$ produced by the latter satisfies $\varphi$. Unfortunately, this approach completely breaks down for an $\FBQP$ machine making a single $\NP$ query, since the query may involve exponentially many inputs in superposition! Can one nevertheless show that $\FNP \subseteq \FBQP^{\NP[1]}$ implies $\NP \subseteq \BQP$?

\paragraph{Organization.} In Section \ref{sec:prelim} we cover definitions and prior results used in this article. In Section \ref{sec:upperbound}, we will give the proof of Theorem \ref{thm:upperbound}, the upper bound on search-to-decision reduction for $\NP$ with quantum access to the oracle. Following that, in Section \ref{sec:lowerbound} we give the proofs of our results in the existential query model, Theorem \ref{blackboxlowerbound} and Corollary \ref{cor:noapxcount}. Finally, we prove Theorem \ref{cor:noapxcount} in Section \ref{sec:whitebox}.

\section{Preliminaries}
\label{sec:prelim}
\subsection{Notation}
Throughout the paper we use $n$ for the number of variables of the formulae and $N = 2^n$ for the size of their truth tables. We write $\sol(\varphi)$ for the set $\{x \in \2^n \colon \varphi(x) \}$ of solutions of the Boolean formula $\varphi$.  If $x \in \2^*$ is a binary string we write $|x|$ for its Hamming weight.

\subsection{Function classes}
We briefly recall the definition the relevant function classes. Contrary to what the name suggests, function classes are actually classes of relations. We will require all relations $R$ in these classes to be \emph{p-bounded}. 
\begin{definition}
    A relation $R \subseteq \2^* \times \2^*$ is called p-bounded if there is some polynomial $p$ such that, for each $x$, if $\exists y . R(x,y)$, then $\exists z$ such that $\len(z) \le p(\len(x))$ and $R(x,z)$. Here $\len(x)$ is the length of the string $x$.
\end{definition}

\begin{definition}
    $\FP$ is the class of poly-time computable\footnote{With a poly-time computable relation we mean that there is a poly-time algorithm for evaluating $R(x,y)$ when given $x$ and $y$ as inputs.} p-bounded relations $R \subseteq \2^* \times \2^*$ such that there is a deterministic poly-time algorithm that on input $x$ does the following:
    \begin{enumerate}
        \item If $\exists y$ such that $R(x,y)$ then the algorithm outputs one such $y$
        \item If $\forall y. (x,y) \notin R$ then the algorithm outputs $\bot$.
    \end{enumerate}
\end{definition}

\begin{definition}
    The class $\FNP$ consists of all poly-time computable p-bounded relations $R \subseteq \2^* \times \2^*$. 
\end{definition}
To see the similarities with the definition of $\FP$ we note that this condition implies the existence of a poly-time non-deterministic algorithm for computing $R$ in the following sense. The algorithm takes as input a string $z$, and each branch of the (non-deterministic) computation outputs either a string $z$ or $\bot$ and satisfies the following properties:
    \begin{enumerate}
        \item If $\exists y$ such that $R(x,y)$, then all branches either output $\bot$ or a string $z$ such that $R(x,z)$ (not necessarily the same one). Furthermore, at least one branch does not output $\bot$.
        \item If $\forall y$ $(x,y) \notin R$, then all branches output $\bot$.
    \end{enumerate}
\noindent An $\FNP$-complete problem is FunctionSAT. It is the relation $R(\varphi, x)$ where $\varphi$ is a (binary encoding of) a Boolean formula and $(\varphi, x) \in R$ iff $x$ is a satisfying assignment of $\varphi$. FunctionSAT is $\FNP$-complete for the same reasons that $\SAT$ is $\NP$-complete.

For the definition of $\FBQP$ we follow Aaronson \cite{aaronson2014equivalence}.
\begin{definition}
    $\FBQP$ is the class of p-bounded relations $R \subseteq \2^* \times \2^*$ for which are computable by a quantum algorithm in the following sense. There exists a poly-time quantum algorithm that takes as input $x$ and $0^{1/\epsilon}$ and outputs a $y$. This is such that $R(x,y)$ with probability at least $1- \epsilon$ (assuming a $y$ with $R(x,y)$ exists). If $\forall y$ $(x,y) \notin R$ then it outputs $\bot$ with probability at least $1 - \epsilon$.
\end{definition}

\subsection{Witness isolation}
We consider algorithms that reduce an arbitrary formula to a uniquely satisfying one.
\begin{definition}[Isolation algorithm] 
    An isolation algorithm with success probability $p$ is a randomized algorithm that maps a Boolean formula $\varphi$ on $n$ variables to a formula $u$ on the same variables such that the formula $\varphi \wedge u$ has a unique solution with probability at least $p$.
\end{definition}
A celebrated result by Valiant and Vazirani states that isolation algorithms exist.
\begin{theorem}[Valiant Vazirani Theorem (\cite{valiant1985np})]
    There exists an isolation algorithm with success probability $\frac{1}{O(n)}$.
\end{theorem}
The main idea of Valiant and Vazirani is to use cleverly chosen pairwise independent hash function to reduce the size of the solutions space.
\begin{definition}(Pairwise independent hash functions (\cite[Definition~8.14]{arora2009computational}))
 A collection $\Hilb_{n,k}$ of functions from $\2^n$ to $\2^k$ is a collection of pairwise independent hash functions if for every $x \ne x' \in \2^n$ and $y, y' \in \2^k$ we have
 \begin{equation}
    \Pr[h(x) = y \wedge h(x') = y'] = 2^{-2k}
 \end{equation}
 where the probability is over $h$ being drawn uniformly at random from $\Hilb_{n,k}$.
\end{definition}

The proof of Valiant and Vazirani's theorem follows from the following lemma which will also be of independent interest for us.
\begin{lem}(\cite[Lemma~17.19]{arora2009computational})
    \label{VVlem}
    Let $\mathcal{H}_{n,k}$ be a collection of pairwise-independent hash functions from $\2^n$ to $\2^k$ and suppose that $\sol(\varphi) \subseteq \2^n$ is such that $2^{k-2} \le |\sol(\varphi)| \le 2^{k-1}$. Then
    \begin{equation}
        \Pr_{h \sim \mathcal{H}_{n,k}}\left[\left|\left\{ x\in \sol(\varphi) : h(x) = 0^k \right\} \right| = 1\right] \ge \frac{1}{8}.
    \end{equation}
\end{lem}

It follows that witness isolation can be performed with constant success probability if the approximate size of the set of solutions is known.

We will also consider witness isolation algorithms with the added requirement that all solutions of $\varphi$ will be the unique solution of $\varphi \wedge u$ with approximately equal probability. We call such algorithms almost-uniform isolation algorithms.
\begin{definition}
    An $\epsilon$-almost-uniform isolation algorithm $A_{iso}$ with success probability $p$ takes as input a Boolean formula $\varphi$, and efficiently produces a Boolean formula $u$ on the same variables such that:
    \begin{itemize}
        \item $p$-Completeness: if $\varphi \in \SAT$, then, with probability at least $p$, $\varphi(x) \wedge u(x)$ has a unique satisfying assignment. In this case we say that the isolation succeeds.
        \item $\epsilon$-almost-uniformity: for all $x \in \sol(\varphi)$ we have: $$ \Pr\Big[\varphi(x) \wedge u(x) \Big| |\sol(\varphi(x) \wedge u(x))| = 1\Big] \le \frac{1+\epsilon}{|\sol(\varphi)|}. $$
    \end{itemize}
    Note that we do not require a lower bound on this probability.
\end{definition}

\subsection{Approximate counting}
Stockmeyer was the first to realize that an $\NP$ oracle can be used for approximate counting. It was shown by Chakraborty, Meel and Vardi that a logarithmic number of $\NP$ queries suffice: \cite{stockmeyer1983complexity,chakraborty2016algorithmic}.
\begin{theorem}[Approximate counting, \cite{stockmeyer1983complexity, chakraborty2016algorithmic}]
    \label{apxcount}
    Given a formula $\varphi$ on $n$ variables and parameters $\delta, \epsilon > 0$, there exists a randomized poly-time algorithm, making $O\left(\frac{\log n \log(1/\delta)}{\epsilon^2}\right)$ queries to an $\NP$-oracle, that outputs a value $c$ such that:
    \begin{equation}
        \Pr\left(\frac{|\sol(\varphi)|}{1 + \epsilon} \le c \le (1 + \epsilon)|\sol(\varphi)|\right) \ge 1-\delta.
    \end{equation}
\end{theorem}

\subsection{Query complexity}
In query complexity one studies how often an algorithm needs to query an input string $x \in \2^N$ in order to compute some function of $x$. In this paper we will consider three different types of oracles: standard oracles, succinct existential oracles and non-succinct existential oracles.

With the \emph{standard oracle model} we will refer to the oracle model that is usually used in quantum query complexity. In this model, the queries give access to a string $x \in \2^N$ using the following query gate:
\begin{equation}
    O_x \colon \ket{i} \mapsto (-1)^{x_i} \ket{i}.
\end{equation}
We will call an application of the oracle \emph{classical} if it is applied to a computational basis state.

In this paper, the string $x$ will usually be the truth table of a hidden formula $\varphi$ (i.e. $x_i = 1 \iff \varphi(i)$). Then, querying the oracle on index $i$ corresponds to computing $\varphi(i)$.

Standard queries do not satisfactorily capture the power of $\NP$ queries. For example, it only takes one $\NP$ query to determine if a formula is satisfiable (i.e. if its truth table is not all 0s), but it is well known that determining if there is an $i$ with $x_i = 1$ takes $\Theta(\sqrt{N})$ standard queries (\cite{grover1996fast,bennett1997strengths}). Therefore, we will also consider other query models that better capture the power of $\NP$ queries.

The first of these query models we will call the \emph{succinct existential query model}. Here, the oracle hides a formula $\varphi$. A query consists of a different poly-size formula $\psi$ and the result to this query will be whether or not $\varphi \wedge \psi$ is satisfiable. Specifically, the oracle can be queried using the query gate
\begin{equation}
    O^{\NP}_\varphi \colon \ket{\psi} \mapsto (-1)^{\sat(\varphi \wedge \psi)}\ket{\psi}
\end{equation}
where $\sat(\varphi \wedge \psi)$ is $1$ if $\varphi \wedge \psi$ is satisfiable and $0$ otherwise.
The succinct $\exists$-query model captures the power of an actual $\NP$ oracle well. It is strong enough to be used for the most common, if not all, well-known algorithms computing properties of Boolean formulae, such as finding the lexicographically-least solution of the hidden formula $\varphi$ (a $\Ptime^{\NP}$-complete problem) and approximately count the number of solutions of $\varphi$ (e.g. with \cite{chakraborty2016algorithmic}).

We will also consider a non-succinct version of the existential query model, which we will also simply call the \emph{existential query model}. Essentially this model is a reformulation of a model originally introduced by Stockmeyer \cite{stockmeyer1983complexity}, but restated in a manner closer to the standard query model. Existential queries ($\exists$-queries) are of the form
\begin{equation}
    O_x^{\exists} \colon \ket{z} \mapsto (-1)^{\overlap(x,z)}\ket{z},
\end{equation}
where $\overlap \colon \2^N \times \2^N \to \2$ is the function
\begin{equation}
    \overlap(x, z) = \begin{cases}
        1 & \text{if } \exists i, x_i = z_i = 1 \\
        0 & \text{otherwise}.
    \end{cases}
\end{equation}
Again, the hidden string will usually be the truth table of a formula $\varphi$. At first glance, this query model may look rather useless. The query register has size exponential in $n$ (the number of variables of $\varphi$). Therefore, even performing a single query will take exponential time. However, \emph{if} one is only concerned with the number of queries, and not with other resources such as time and space, then non-succinct existential queries are more powerful than succinct ones. Instead of making a succinct $\exists$-query with formula $\psi$, an algorithm can make a non-succinct existential query with $z$ the truth table of $\psi$. Furthermore, not all truth tables correspond to poly-size formulae. In this paper we will prove lower bounds on the number of existential queries an algorithm needs to make. The result will allow the algorithms unbounded time and space and hence these lower bounds will in particular hold for efficient algorithms making succinct $\exists$-queries.

We are interested in the number of these queries needed to find a solution of $\varphi$. In the existential query model this corresponds to solving the search problem on the truth table $x$ of $\phi$, that is, outputting an index $i$ such that $x_i = 1$. It should be noted that, unlike in the case of standard queries, the existential query complexity of this function search problem is not necessarily the same as that of the decision search problem (i.e. determining if there is such an index). For example, an information theoretic argument shows that classically $\Theta(n)$ existential queries are needed for function search, but a single existential query with $z = 1^N$ solves decision search.\footnote{The string $1^N$ has overlap with any non-zero string. Hence it has overlap with the truth table of $\varphi$ iff $\varphi$ is satisfiable.}

We will also be interested in a slight variation of the search problem which we call the \emph{index sampling} problem. Here the task is to sample according to the uniform distribution on the support of $x$. Formally, we define it as
\begin{definition}[Index sampling]
    An algorithm solves the index sampling problem if it, for all $x \in \2^N \backslash \{0^N\}$, outputs $s \in [N] \cup \{\bot\}$ such that:
    \begin{itemize}
        \item for all $i \in [N]$ with $x_i = 1$, $\Pr[s = i | s \neq \bot] = \frac{1}{|x|}$,
        \item there is a constant $c$ such that $\Pr[s \neq \bot] \ge c$.
    \end{itemize}

\end{definition}

\section{Quantum algorithm for search-to-decision reduction}
\label{sec:upperbound}
We are now ready to prove our results. For pedagogical reasons we start with the proof of the upper bound (Theorem \ref{thm:upperbound}) before proving our main results: the lower bounds in Theorem \ref{blackboxlowerbound} and Theorem \ref{cor:noapxcount}.
\begin{reptheorem}{thm:upperbound}
    $\FNP \subseteq \FBQP^{\NP[\log]}$. Furthermore, all queries made to the oracle are of the form $\varphi \wedge \chi$ where $\varphi$ is the input formula and $\chi$ some other formula.
\end{reptheorem} 

\begin{proof}
    We will show that there exists an $\FBQP^{\NP[\log]}$ algorithm that, when given a $\SAT$ instance $\varphi$, outputs a satisfying assignment $x\in \2^n$ of $\varphi$ if one exists, and outputs ``no solution'' otherwise. The algorithm succeeds with constant probability. This success probability can be boosted by running the algorithm a constant number of times, checking for each output if it is indeed a satisfying assignment and then outputting one that is. Hence, the success probability can be taken to be any arbitrary constant.

    The existence of a satisfying assignment can be checked using a single query to the $\NP$-oracle. Therefore, we will restrict our attention to the case where a satisfying assignment exists. We will proceed in two steps. First, we show how the satisfying assignment of a formula with \emph{exactly} one satisfying assignment can be found with only a single query to the $\NP$-oracle using the Bernstein Vazirani algorithm. Then we show how we can use $O(\log n)$ queries to reduce any formula to a uniquely satisfying one with constant probability.

    \begin{lem}[``Bernstein-Vazirani (BV) trick'' \cite{irani2021quantum}]
        \label{iranilemma}
        Let $\varphi$ be a formula with exactly one solution. There exists a $\BQP$ algorithm that makes a single query to an $\NP$-oracle and finds this unique solution with probability 1.
    \end{lem}
    \begin{subproof}[Proof of BV trick]
        Let the unique solution of $\varphi$ be denoted by $s$ and consider the formula 
        \begin{equation}
            \psi_a \coloneqq \varphi(x) \wedge (x\cdot a = 1),
        \end{equation} 
        where $a\in \2^n$ and $x\cdot a$ denotes the inner product of the two binary strings $x$ and $a$ modulo $2$. (Note that $\psi_a$ is of the form $\varphi \wedge \chi$.) We now have that $\psi_a$ is satisfiable (i.e. there is a $y$ such that $\varphi(y) \wedge y\cdot a = 1$) if and only if\footnote{Note this does not imply $a = s$.} $a\cdot s = 1$ because $s$ is, by assumption, the only solution of $\varphi$. Now we run the Bernstein-Vazirani algorithm, where to evaluate $a\cdot s$ we ask the $\NP$ machine whether $\psi_a$ is true. That is, we start with the state $\ket{0^n}$ and apply $H^{\otimes n}$ to get the uniform superposition on $n$ qubits. The next step is to query the oracle on input $\psi_a$ where $a$ is in a uniform superposition:
        \begin{equation}
            \frac{1}{\sqrt{2^n}}\sum_{a \in \2^n} \ket{a} \mapsto \frac{1}{\sqrt{2^n}} \sum_{a \in \2^n} (-1)^{a\cdot s}\ket{a}.
        \end{equation}
        Now another application of $H^{\otimes n}$ gives:
        \begin{align}
            \frac{1}{2^n}\sum_{a,y \in \2^n} (-1)^{a\cdot s + a\cdot y} \ket{y} = \frac{1}{2^n}\sum_{a,y \in \2^n} (-1)^{a\cdot(s \oplus y)} \ket{y} = \ket{s}.
        \end{align}
    Hence measuring the final state in the computational basis gives the unique satisfying assignment $s$ of $\varphi$.
    \end{subproof}
    To deal with an arbitrary number of solutions, we first use Theorem \ref{apxcount} and $O(\log n)$ queries to the $\NP$ oracle to find $k$ such that $2^{k-2} \le |\sol(\varphi)| < 2^{k-1}$. (All queries made by the approximate counting algorithm in \cite{chakraborty2016algorithmic} are of the form $\varphi \wedge \chi$.) Then, we invoke Lemma \ref{VVlem} to obtain $u(x) \coloneqq h(x) = 0^k$ such that $\varphi \wedge u$ has a unique solution with probability $>\frac{1}{8}$. Applying the BV trick to $\varphi \wedge u$ then completes the proof. Furthermore, all queries made were of the claimed form.
\end{proof}

\section{Lower bound for existential query complexity of search}
\label{sec:lowerbound}
We will prove that all quantum algorithms for the search problem need $\Omega(\log n)$ existential queries, even if we allow the algorithm to additionally make $\poly(n)$ classical standard queries. To do so, we will reduce binary search to this problem in order to use Ambainis' lower bound for binary search \cite{ambainis1999better}. A binary search problem consists of a monotonic binary string $x = 00\dots 01 \dots 1$ and the task is to find the index of the first $1$.

\begin{reptheorem}{blackboxlowerbound}[Restated]
    Any quantum algorithm with existential query access to $x\in \2^N$ needs to make $\Omega(\log\log N) = \Omega(\log n)$ existential queries to find an $i$ such that $x_i = 1$. This remains true even if the algorithm is allowed to make an additional $\poly(n)$ classical standard queries.
\end{reptheorem}
The proof will follow from the following lemma.
\begin{lem}
\label{lowerboundlem}
    Consider the following statements:
    \begin{enumerate}
        \item There exists a quantum algorithm for search on strings of size $N = 2^n$ that makes $q$ $\exists$-queries and $\poly(n)$ \emph{classical} standard queries and succeeds with constant probability.
        \item There exists a quantum algorithm for search on strings of size $N$ using $q + O(\log \log n)$ $\exists$-queries and no additional standard queries which succeeds with constant probability.
        \item There exists a quantum algorithm for index sampling on strings of size $N$ using $q + O(\log \log n)$ $\exists$-queries and no additional standard queries which succeeds with constant probability.
        \item There exists a quantum algorithm for binary search on strings of size $n$ using $q + O(\log \log n)$ $\exists$-queries and no additional standard queries. The algorithm succeeds with constant probability.
        \item There exists a quantum algorithm for binary search on strings of size $n$ using $q + O(\log \log n)$ standard quantum queries and no $\exists$-queries. The algorithm succeeds with constant probability.
    \end{enumerate}
    Then, $1 \implies 2 \implies 3 \implies 4 \implies 5$.
\end{lem}
\noindent It is worth noting that all success probabilities can be boosted to be bigger than any constant $c<1$.
\begin{proof}
$1\implies 2$:
consider an algorithm for search using $q$ $\exists$-queries and $\text{poly}(n)$ classical standard queries. We will modify the algorithm to get rid of the standard queries. Our new modified algorithm will act exactly the same as the original algorithm, except it does not actually perform the classical standard queries. Instead, it assumes the answer to those queries is $0$ and keeps track of the positions that should have been queried in a set $A$. At the end of the algorithm, it checks if its assumptions were correct using an $\exists$-query. That is, it performs an $\exists$-query with string $z$ defined by $z_i = 1$ iff the $i$-th index should have been queried by a classical standard query at some point.

Now there are now two options. Either the result of this $\exists$-query is 0, in which case the assumptions that all indices queried by classical standard queries were $0$ is correct, and hence the modified algorithm and the original algorithm coincide.
Alternatively, the result of the $\exists$-query is $1$. Then at least one of the assumptions was wrong. But now the algorithm has determined that the set $A \subseteq [N]$ contains some $i \in A$ with $x_i = 1$. Furthermore, $A$ contains at most $\poly(n)$ elements since only $\poly(n)$ classical queries were made. The search algorithm from Theorem \ref{thm:upperbound} makes only (succinct) existential queries, so we can now use it to search, within $A$, for an $i$ with $x_i = 1$.\footnote{We can add $\wedge x \in A$ to all formulae to restrict the search to within $A$.} Because we know that $|A| = \poly(n)$, approximately counting the number of solutions \emph{within $A$} will take $O(\log\log |A|) =  O(\log \log n)$ $\exists$-queries using the algorithm from \cite{chakraborty2016algorithmic}.\footnote{Essentially, the algorithm from \cite{chakraborty2016algorithmic} uses binary search to find $k\in[n]$ such that $2^{k-1} < |\sol(\varphi)| < 2^k$. Because $|A| = O(\poly(n))$, it is already known that $|\sol(\varphi(x) \wedge x \in A)| \le O(\poly(n)) = 2^{O(\log n)}$. Therefore, the binary search is sped up.} Therefore, finding a solution within $A$ given that there is one will take $O(\log \log n)$ existential queries.

$2 \implies 3$: consider a permutation $\sigma \in S_N$ mapping $[N]$ to itself drawn uniformly at random. With slight abuse of notation we define $\sigma(x)$ by $\sigma(x)_i = x_{\sigma(i)}$. Our algorithm for index sampling will apply the algorithm for search to $\sigma(x)$, undo the permutation, check if it is indeed a solution and output the result if it is, and abort (i.e. output $\bot$) if it is not. The probability of aborting is exactly the failure probability of the search algorithm. In the following we condition on the sampling algorithm not aborting and assume $x \neq 0^N$ (this case can be checked with 1 $\exists$-query).

Denote by $\search(x)$ the output of the search algorithm on input $x$. Consider the probability $p(i) = \Pr[\sigma^{-1}(\search(\sigma(x))) = i] = \Pr[\search(\sigma(x)) = \sigma(i)]$. We will show that $p(i) = \frac{1}{|x|}$ if $x_i = 1$ and $0$ otherwise. Note that there are two sources of randomness: the random choice of $\sigma$ and potentially random behavior of the search algorithm. We can write:
\begin{align}
    p(i) &= \sum_{y \in \2^N} \sum_{k \in [N]} \Pr[\sigma(x) = y \wedge \sigma(i) = k \wedge \search(y) = k] \\
    &= \sum_{y,k} \Pr[\sigma(x) = y]\cdot\Pr[\sigma(i) = k | \sigma(x) = y]\cdot\Pr[\search(y) = k | \sigma(x) = y \wedge \sigma(i) = k].
\end{align}
Because the algorithm for search does not depend on $\sigma$, we have that
\begin{align}
    p_y(k) \coloneqq \Pr[\search(y) = k | \sigma(x) = y \wedge \sigma(i) = k] = \Pr[\search(y) = k].
\end{align}
The $p_y(k)$ are unknown, but because the algorithm solves the search problem we do know that $\Pr[\search(y) = k | y_k = 0] = 0$ and $\sum_{k: y_k = 1} p_y(k) = 1$. Furthermore, we have
\begin{equation}
    \Pr[\sigma(x) = y] = \begin{cases} 1/{N\choose{|x|}} &\text{ if } |x| = |y|  \\ 0 &\text{ if } |x| \ne |y| \end{cases}
\end{equation}
and
\begin{equation}
    \Pr[\sigma(i) = k | \sigma(x) = y] = \begin{cases}
        \frac{1}{|x|} & \text{ if } x_i = y_k = 1 \\
        \frac{1}{N - |x|} & \text{ if } x_i = y_k = 0\\
        0 & \text{ if } x_i \ne y_k.
    \end{cases}
\end{equation}
We can now put everything together to get
\begin{align}
    p(i) &= \sum_{y : |x| = |y|} \frac{1}{{N\choose{|x|}}} \sum_{k : y_k = 1} \frac{1}{|x|} p_y(k) = \frac{1}{|x|}
\end{align}
if $x_i = 1$. On the other hand, if $x_i = 0$, we have that $p_y(k) = 0$ if $y_k = 0$ and that $\Pr[\sigma(i) = k | \sigma(x) = y] = 0$ if $y_k = 1$. Therefore, we have $p(i) = 0$ if $x_i = 0$.

$3 \implies 4$: let $y \in \2^{n}$ be a binary search instance. That is, $y = 00\dots011\dots1$ is monotonically increasing. For binary search, we want to find the smallest index $i$ such that $y_i = 1$. We will now define another (exponentially longer) binary string $x \in \2^N$ by $x_1 = x_2 = \dots = x_{N/2} = y_1$, $x_{N/2 + 1} = \dots = x_{3N/4} = y_2$ and so on, ending with $x_{N-2} = x_{N-1} = y_{\log N - 1}$ and $x_n = y_{\log N}$, i.e.
\begin{equation}
    x = \underbrace{y_1 y_1 \dots y_1}_{\frac{N}{2} \text{ times}} \underbrace{y_2 y_2 \dots y_2}_{\frac{N}{4} \text{ times}} \dots \underbrace{y_i y_i \dots y_i}_{\frac{N}{2^i} \text{ times}} \dots y_{\log N - 1} y_{\log N -1} y_{\log N}.
\end{equation}
If we sample an $i$ with $x_i = 1$ uniformly at random, it will correspond to the smallest $j$ with $y_j = 1$ with probability $>\frac{1}{2}$. This is because by construction each $y_j$ appears more in $x$ than all $y_k$ for $k>j$ together. Furthermore, each query to $y$ (standard or existential) can be simulated by a single query of the same kind to $x$. Therefore, statement $3.$ allows us to solve binary search with constant success probability and $q + O(\log\log n)$ $\exists$-queries.

\begin{figure}[h!]
\begin{mdframed}[align = center, linecolor = black!100, linewidth = .3mm, skipbelow = 0cm, skipabove=0mm]
\begin{center}
\begin{tikzpicture}[
spacednode/.style={circle, very thick, minimum size=3mm},
]
\node      (y1)                              {$y_1$};
\node[spacednode]      (top)       [left=of y1]           {};
\node      (y2)            [right=5mm of y1] {$y_2$};
\node      (yn-1)            [right= 14mm of y2] {$y_{n-1}$};
\node      (yn)            [right=5mm of yn-1] {$y_{n}$};
\node      (xN/2)            [below=14mm of y1] {$x_{N/2}$};
\node      (x1)            [left=14mm of xN/2] {$x_{1}$};
\node      (xN/2+1)            [right=5mm of xN/2] {$x_{N/2 + 1}$};
\node      (x3N/4)            [right=14mm of xN/2+1] {$x_{3N/4}$};
\node      (xN-2)            [right=14mm of x3N/4] {$x_{N-2}$};
\node      (xN-1)            [right=5mm of xN-2] {$x_{N-1}$};
\node      (xN)            [right=5mm of xN-1] {$x_N$};
\node       (spacer)        [above=7.5mm of xN-1]      {};
\node       (text1)         [right=14mm of spacer]       {\textbf{Duplicate}};
\node       (text2)         [below=14mm of text1]    {\textbf{Permute}};
\node      (sx1)            [below=14mm of x1] {$\sigma(x)_1$};
\node      (sxN)            [below=14mm of xN] {$\sigma(x)_N$};
\node      (empty1)            [below=14mm of xN/2] {};
\node      (empty2)            [below=14mm of xN/2+1] {};
\node      (empty3)            [below=14mm of x3N/4] {};
\node      (empty4)            [below=14mm of xN-2] {};
\node      (empty5)            [below=14mm of xN-1] {};

\draw[loosely dotted, very thick] (y2.east) -- (yn-1.west);
\draw[-{Latex},very thin] (y1.south) -- (x1.north);
\draw[-{Latex}] (y1.south) -- (xN/2.north);
\draw[-{Latex}] (y2.south) -- (xN/2+1.north);
\draw[-{Latex}] (y2.south) -- (x3N/4.north west);
\draw[-{Latex}] (yn-1.south) -- (xN-2.north west);
\draw[-{Latex}] (yn-1.south) -- (xN-1.north west);
\draw[-{Latex}] (yn.south) -- (xN.north west);
\draw[loosely dotted, very thick] (x1.east) -- (xN/2.west);
\draw[loosely dotted, very thick] (xN/2+1.east) -- (x3N/4.west);
\draw[loosely dotted, very thick] (x3N/4.east) -- (xN-2.west);
\draw[-{Latex}] (x1.south east) -- (empty4.north);
\draw[-{Latex}] (xN/2.south) -- (empty2.north);
\draw[-{Latex}] (xN/2+1.south) -- (sx1.north east);
\draw[-{Latex}] (x3N/4.south) -- (sxN.north west);
\draw[-{Latex}] (xN-2.south) -- (empty5.north);
\draw[-{Latex}] (xN-1.south) -- (empty1.north);
\draw[-{Latex}] (xN.south west) -- (empty3.north);
\draw[loosely dotted, very thick] (sx1.east) -- (sxN.west);

\end{tikzpicture}
\end{center}
\end{mdframed}
\vspace*{-5mm}
\caption*{Modification of the binary-search oracle $y$. Queries to $\sigma(x)$ are made by ``going up the arrows''.}
\end{figure}
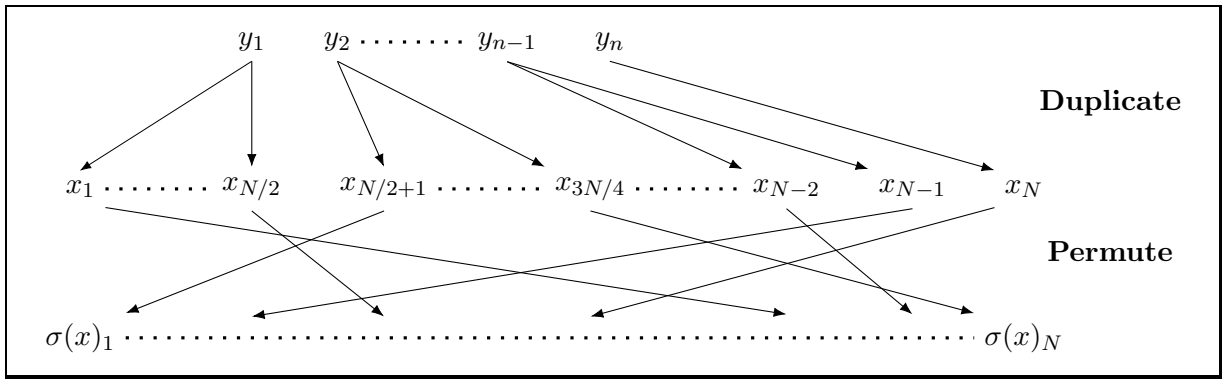

$4 \implies 5$: because the strings considered in a binary search problem are monotonically increasing we can simulate an $\exists$-query on such a string using only a single standard query. Instead of an $\exists$-query with $z$, we find the largest $i$ such that $z_i = 1$ and use a standard query to query the $i$-th bit. The claimed implication follows.
\end{proof}
With this lemma in hand, Theorem \ref{blackboxlowerbound} is easily proven.
\begin{proof}[Proof of Theorem \ref{blackboxlowerbound}]
    Suppose there exists an algorithm for search on $x \in \2^N$ making $q = o(\log n)$ existential queries and $\poly(n)$ classical standard queries. Then, by Lemma \ref{lowerboundlem} there is an algorithm for binary search on strings of length $n$ making $q + O(\log\log n) = o(\log n)$ standard queries. This contradicts Ambainis' lower bound on binary search, which states that $\Omega(\log n)$ queries are required for binary search on size $n$ strings \cite{ambainis1999better}. Therefore, any search algorithm needs to make at least $q = \Omega\log n$ existential queries, even if it also makes $\poly(n)$ classical standard queries.
\end{proof}

Corollary \ref{cor:blackboxnoapxcount} is an easy consequence of the previous theorem.

\begin{repcorollary}{cor:blackboxnoapxcount}
    Any quantum algorithm that is given existential query access to a string $x \in \2^N$ and outputs an estimate $c$ such that $2^{|x| - 1} \le c < 2^{|x|}$, where $|x|$ is the Hamming weight of $x$, needs to make at least $\Omega(\log\log N)$ queries to the oracle.
\end{repcorollary}
\begin{proof}
    We show that the existence of such an approximate counter making $o(\log \log N) = o(\log n)$ $\exists$-queries implies the existence of an algorithm for search making $o(\log n)$ $\exists$-queries. Using the approximate counter and Lemma \ref{VVlem} we can reduce $x$ to be of Hamming weight 1 with constant probability and $o(\log n)$ $\exists$-queries. The index of the unique $1$ can then be extracted using the BV trick (Lemma \ref{iranilemma}) and a single $\exists$-query.
\end{proof}

\section{Conditional lower bound on number of $\NP$ queries for approximate counting}

\label{sec:whitebox}
In order to prove Corollary \ref{cor:noapxcount} we will first concern ourselves with almost-uniform isolation algorithms. We will first prove the following theorem stating consequences of the existence of almost-uniform isolation algorithms making $o(\log n)$ $\NP$ queries. Thereafter, we will show that approximate counting with $o(\log n)$ $\NP$ queries implies almost-uniform isolation with $o(\log n)$ $\NP$ queries.
\begin{theorem}
    Let $\epsilon > 0$ and $p \in (0,1]$ be constants such that there exists an $\epsilon$-almost-uniform isolation algorithm $A_{iso}$ with success probability $p$ making $Q(n) = o(\log n)$ queries to an $\NP$-oracle. Then, $\BPP^{\NP\left[o(n)\right]}$ contains a $\Ptime^{\NP}$-complete problem if this algorithm is classical and $\FBQP^{\NP\left[o(n)\right]}$ contains an $\FP^{\NP}$-complete problem if it is quantum.
\end{theorem}
\begin{proof}
    We will give a $\BPP^{\NP[o(n)]}$ algorithm to isolate the lexicographically least solution of an input formula $\varphi$, as outputting the last bit of this solution is $\Ptime^{\NP}$-complete \cite{krentel1986complexity}. The idea of the algorithm will be as follows. We work in rounds. In the first round we sample $n^2$ almost-uniformly random solutions of $\varphi$. We will not explicitly know what these solutions are, but we can find a formula to which they are the unique solution. Next, we sample $n^2$ solutions of $\varphi$ among all solutions that are lexicographically smaller than the all previously found solutions. We keep going like this until only one solution remains, which will be the lexicographically least solution of $\varphi$. We show that with high probability, the number of solutions that are smaller than the least solution found yet decreases by at least a factor $\frac{1}{n}$ every round. Therefore, we will, with high probability, need at most $\log_n(|sol(\varphi)|) = \frac{\log(|sol(\varphi)|)}{\log n} \le \frac{n}{\log n}$ rounds to isolate the lexicographically least solution of $\varphi$.

    In the first round we apply $A_{iso}$ to $\Psi_1(\vec{x}_1, \dots, \vec{x}_{n^2}) \coloneqq \varphi(\vec{x}_1) \wedge \dots \wedge \varphi(\vec{x}_{n^2})$. Here the $\vec{x}_i$ denote fresh sets of variables and we use the vector notation to emphasize that they are $n$-bit strings and not bits. The result of this application will be $u_1(\vec{x}_1, \dots, \vec{x}_{n^2})$ such that $\Psi_{1} \wedge u_1$ has a unique solution. This unique solution will be the concatenation of $n^2$ solutions of $\varphi$.

    In round $r+1$ we will do the following. From the previous round we have already constructed $$\Psi_r(\vec{x}_{r,1}, \dots, \vec{x}_{r,n^2}, \dots, \vec{x}_{1,1}, \dots, \vec{x}_{1,n^2})$$ with a unique solution. In this unique solution, $\vec{x}_{i,1}, \dots, \vec{x}_{i, n^2}$ will be the solutions to $\varphi$ picked in round $i$. We set
    \begin{align}
        \label{bigeq}
        \begin{split}
            \chi_{r+1}(\vec{x}_{r+1,1}, \dots, \vec{x}_{r+1,n^2}, \dots, \vec{x}_{1,1}, \dots, \vec{x}_{1,n^2}) &\coloneqq \varphi(\vec{x}_{r+1, 1}) \wedge \dots \wedge \varphi(\vec{x}_{r+1, n^2}) \\ 
            & \wedge \Psi_{r}(\vec{x}_{r,1}, \dots, \vec{x}_{r,n^2}, \dots, \vec{x}_{1,1}, \dots, \vec{x}_{1,n^2}) \\
            & \wedge \vec{x}_{r+1, 1} <_{lex} \vec{x}_{r,1} \wedge \dots \wedge \vec{x}_{r+1,1} <_{lex} \vec{x}_{r,n^2} \\
            & \wedge \vec{x}_{r+1, n^2} <_{lex} \vec{x}_{r,n^2} \wedge \dots \wedge \vec{x}_{r+1,1} <_{lex} \vec{x}_{r,n^2}.
        \end{split}
    \end{align}
    In any satisfying assignment of $\chi_{r+1}$, the first line of the RHS enforces that $\vec{x}_{r+1,1}, \dots, \vec{x}_{r+1,n^2}$ are solutions of $\varphi$. The second line makes sure that $\vec{x}_{r,1}, \dots, \vec{x}_{r,n^2}, \dots, \vec{x}_{1,1}, \dots, \vec{x}_{1,n^2}$ are set to the unique solutions picked in previous rounds. The third and fourth lines make sure that the new solutions are lexicographically strictly smaller than any solution picked in a previous round.\footnote{Note that comparing only to the previous rounds solutions suffices.}

    We now pick the new solutions of round $r+1$ by applying $A_{iso}$ to $\chi_{r+1}$. We call the round is successful if $A_{iso}$ succeeds, i.e. if its output $u_{r+1}$ is such that $\chi_{r+1} \wedge u_{r+1}$ has a unique solution. We can check if $A_{iso}$ succeeded by spending 2 $\NP$ queries.\footnote{One query is used to check if the formula is satisfiable and the other is used to check if there are two or more distinct solutions.} In case of a success $\Psi_{r+1} \coloneqq \chi_{r+1} \wedge u_{r+1}$ will have a unique solution, which is picked almost uniformly at random from all solutions of $\chi_{r+1}$ (by definition of $A_{iso}$). 
    In this unique solution $\vec{x}_{r+1,1}, \dots, \vec{x}_{r+1, n^2}$ will be the newly picked solutions. By construction they will be smaller than the solutions found in the previous rounds. The previously picked solutions $\vec{x}_{r,1}, \dots, \vec{x}_{r,n^2}, \dots, \vec{x}_{1,1}, \dots, \vec{x}_{1,n^2}$ will be the same as in previous rounds because the second line of Equation \ref{bigeq} has a unique solution. 
    Finally, we check if there are still smaller solutions available by checking if
    \begin{equation}
        \varphi(y) \wedge \Psi_{r+1}(\vec{x}_{r+1,1}, \dots, \vec{x}_{r+1,n^2}, \dots, \vec{x}_{1,1}, \dots, \vec{x}_{1,n^2}) \wedge y <_{lex} \vec{x}_{r+1, 1} \wedge \dots \wedge y <_{lex} \vec{x}_{r+1, n^2}
    \end{equation}
    is still satisfiable. If it is satisfiable we proceed to the next round and if it is not then, in the unique solution of $\Psi_{r+1}$, one of the $\vec{x}_{r+1, i}$ will be the minimal solution of $\varphi$. Therefore, 
    \begin{equation}
        \varphi(y) \wedge \Psi_{r+1}(\vec{x}_{r+1,1}, \dots, \vec{x}_{r+1,n^2}, \dots, \vec{x}_{1,1}, \dots, \vec{x}_{1,n^2}) \wedge y \le_{lex} \vec{x}_{r+1, 1} \wedge \dots \wedge y \le_{lex} \vec{x}_{r+1, n^2}
    \end{equation}
    will have a unique solution where $y$ is the lexicographically least assignment of $\varphi$ (note the use of $\leq$ instead of $<$). Asking the oracle if this formula is still satisfiable with the last bit of $y$ set to $1$ will then tell us the last bit of the lexicographically least solution of $\varphi$. Alternatively, a $\BQP$-machine can use the BV trick to obtain the \emph{entire} lexicographically least solution with one query and solve the $\FP^{\NP}$-complete problem of outputting the lexicographically least solution of $\varphi$ \cite{krentel1986complexity}.
    \\
    We will proceed by proving that this algorithm succeeds with probability at least $\frac{2}{3}$ and makes at most $o(n)$ $\NP$ queries using the following claims:

    \begin{claim}[Number of successful rounds needed]
        For sufficiently large $n$, the probability that the algorithm described above has not terminated after $\frac{n}{\log n}$ \emph{successful} rounds is less than $\frac{1}{6}$ (a successful round is a round in which a unique solution remains after the application of $A_{iso}$).
    \end{claim}

    \begin{claim}[Probability of successful rounds]
        The probability that after $\frac{2n}{p \log n}$ rounds there have not been $\frac{n}{\log n}$ \emph{successful} rounds is at most $\frac{1}{6}$ for sufficiently large $n$. Here, recall $p$ is the success probability of $A_{iso}$.
    \end{claim}

    From these claims it follows that with probability at least $\frac{2}{3}$, the algorithm will, after at most $\frac{2n}{p\log n}$ rounds, have terminated. In every round, $Q(\poly(n))$ queries are made because only a polynomial amount of terms are added to $\Psi_r$ every round (recall that $Q(n)$ is the number of queries made by $A_{iso}$). Since $Q(n) = o(\log n)$ by assumption, we have $Q(n^c) = o(c\cdot \log n) = o(\log n)$. Hence the algorithm makes at most $\frac{2n}{p\log n} o(\log n) = o(n)$ queries.

    Finally, we prove the two claims.

    \begin{subproof}[Proof of claim 1]
        Let $\vec{y}_1 < \dots < \vec{y}_k$ denote all solutions to $\varphi$ smaller than $\vec{y}_{min,r-1}$, the smallest solution found in round $r-1$. Define $g = \ceil*{\frac{k}{n}}$. We now compute the probability that, in a single round $r$, $\vec{y}_{min, r}$ is larger than $\vec{y}_{g}$ as
        \begin{align}
            \Pr[\vec{y}_{min, r} >_{lex} \vec{y}_{g}] &= \Pr\left[\vec{x}_{r,1} >_{lex} \vec{y}_{g} \wedge \dots \wedge \vec{x}_{r,n^2} >_{lex} \vec{y}_{g}\right] \\
            &= \Pr[(\vec{x}_{r,1}, \dots, \vec{x}_{r,n^2}) \in \{\vec{y}_{g+1}, \dots, \vec{y}_{k}\}^{n^2}] \\
            &\le |\{\vec{y}_{g+1}, \dots, \vec{y}_{k}\}^{n^2}|\cdot \frac{1 + \epsilon}{|\{\vec{y}_1, \dots, \vec{y}_k\}^{n^2}|} \\ \label{unionboundenepsiso}
            &= (k-g)^{n^2} \frac{1+\epsilon}{k^{n^2}} \\
            &\le \left(k\left(1-\frac{1}{n}\right)\right)^{n^2} \frac{1+\epsilon}{k^{n^2}}\\
            &= (1+\epsilon)\left(1-\frac{1}{n}\right)^{n^2},
        \end{align}
        where Equation \ref{unionboundenepsiso} follows from a union bound and the definition of an $\epsilon$-almost-uniform isolation algorithm
        Hence, by a union bound, the probability that in at least one of $\frac{n}{\log n }$ successful rounds $\vec{y}_{min,r} >_{lex} \vec{y}_{goal,r}$, i.e., the probability that in at least on of the rounds the search space is not cut down by at least a factor $\frac{1}{n}$ is
        \begin{equation}
            \Pr[\exists r \le \frac{n}{\log n} \text{ s.t. } \vec{y}_{min,r} >_{lex} \vec{y}_{g,r}] \le (1+\epsilon) \frac{n}{\log n}\left(1 - \frac{1}{n} \right)^{n^2},
        \end{equation}
        of which the right-hand side goes to $0$ as $n$ goes to $\infty$. The claim follows.
    \end{subproof}

    \begin{subproof}[Proof of claim 2]
        Each round succeeds with probability $p$. After $\frac{2n}{p\log n}$ rounds the expected number of successful rounds is $\frac{2n}{\log n}$. By a Chernoff bound we have:
        \begin{equation}
            \Pr[\# \text{successes} < \frac{n}{\log n}] \le \exp\left(-\frac{n}{4p\log n}\right).
        \end{equation}
        For sufficiently large $n$, the right-hand side will indeed be at most $\frac{1}{6}$.
    \end{subproof}

\end{proof}

Corollary \ref{cor:noapxcount} follows by showing that an algorithm for approximate counting induces an algorithm for almost-uniform isolation with the same number of $\NP$ queries.

\begin{repcorollary}{cor:noapxcount}
    If there exists a classical randomized poly-time algorithm for approximate counting making $o(\log n)$ $\NP$ queries, then $\BPP^{\NP[o(n)]}$ contains a $\Ptime^{\NP}$-complete problem. Similarly, if there is a poly-time quantum algorithm for approximate counting making $o(\log n)$ $\NP$ queries, then $\BQP^{\NP[o(n)]}$ contains an $\FP^{\NP}$-complete problem.
\end{repcorollary}
\begin{proof}
    The approximate counting algorithms can be used to make an almost-uniform isolation algorithm as follows. First, use the approximate counting algorithm to find $k$ such that $2^{k-2} \le |\sol(\varphi)| \le 2^{k-1}$. Next, choose a random hash function $h$ from a set of pairwise independent hash functions from $\2^n$ to $\2^k$. By Lemma \ref{VVlem}, the formula $\varphi(x) \wedge h(x) = 0^k$ will then have a unique solution with probability at least $\frac{1}{8}$.\footnote{We only require almost-uniformity in the case that there is a unique solution.} We claim that in this case the unique solution will be distributed almost-uniformly at random among all solutions of $\varphi$. Our proof of this claim is based on work by Dellanoy and Meel \cite[Lemma~3]{delannoy2022almost}.

    Before we prove the claim let us first introduce some notation. Let the random variable $N$ denote the number of solutions of $\varphi(x) \wedge h(x) = 0^k$ and let $SC$ denote the event that the approximate counting was successful, i.e. $2^{k-2} \le |\sol(\varphi)| \le 2^{k-1}$. We assume that $SC$ occurs with probability $1 - \delta$. For any fixed $x$, we are interested in
    \begin{align}
        \Pr[h(x) = 0^k| N = 1] &= \frac{\Pr[h(x)=0^k \wedge N = 1]}{\Pr[N = 1]} \\
        &= \frac{\Pr[h(x)=0^k \wedge N = 1]}{\sum_{y \in \sol(\varphi)} \Pr[h(y) = 0^k \wedge N = 1]} \\
        &= \frac{\Pr[N = 1|h(x) = 0^k] \Pr[h(x) = 0^k]}{\sum_{y \in \sol(\varphi)} \Pr[N = 1|h(y) = 0^k]\Pr[h(y) = 0^k]} \\
        &= \frac{\Pr[N = 1|h(x) = 0^k]}{\sum_{y \in \sol(\varphi)} \Pr[N = 1|h(y) = 0^k]} \label{22}
    \end{align}
    and want to show that this probability is $\le \frac{\epsilon}{|\sol(\varphi)|}$.
    Let $\chi(y)$ be $1$ if $h(y) = 0^k$ and $0$ otherwise. For any $x \in \sol(\varphi)$ we have:
    \begin{align}
        \mathbb{E}\left[N\middle| h(x) = 0^k \wedge SC\right] &= \mathbb{E}\left[\sum_{y \in \sol(\varphi)} \chi(y) \middle| h(x) = 0^k \wedge SC\right] \\
        &= \sum_{y \in \sol(\varphi)} \mathbb{E}\left[\chi(y) \middle| h(x) = 0^k \wedge SC\right] \\
        &= 1 + \frac{|\sol(\varphi)| - 1}{2^{-k}}
    \end{align}
    where we use that $h$ is a $2$-wise independent hash function. By Markov's inequality this means that for all $x \in \sol(\varphi)$
    \begin{align}
        \Pr[N = 1| h(x) = 0^k \wedge SC] &= 1 - \Pr[N \ge 2| h(x) = 0^k \wedge SC] \\
        &\ge 1 - \frac{\mathbb{E}[N| h(x) = 0 \wedge SC]}{2} \\
        &\ge \frac{1}{2} - \frac{|\sol(\varphi)| - 1}{2^{k+ 1}} \\
        &\ge \frac{1}{4}.
    \end{align}
    Combining with Equation \ref{22} and using that $\Pr[N = 1 | h(y) = 0^k] \ge \Pr[N = 1 \wedge SC | h(y) = 0^k]$ and $\Pr[N = 1|h(x) = 0^k] \le 1$ gives
    \begin{align}
        \Pr[h(x) = 0^k| N = 1] &\le \frac{1}{\sum_{y \in \sol(\varphi)} \Pr[N = 1 \wedge SC|h(y) = 0^k]} \\
        &\le \frac{1}{\sum_{y \in \sol(\varphi)} \Pr[N = 1|h(y) = 0^k \wedge SC]\Pr[SC]} \\
        &\le \frac{4}{(1-\delta)|\sol(\varphi)|}.
    \end{align}
    Note that this holds for \emph{any} $x$. Hence the existence of an approximate counter gives us an almost-uniform isolation algorithm with $p = \frac{1-\delta}{8}$ and $\epsilon = \frac{3 + \delta}{1-\delta}$. Furthermore, the approximate counter and isolation algorithm will make the same amount of $\NP$ queries (i.e. $o(\log n)$).
\end{proof}

\section{Acknowledgements}
The authors would like to thank Ronald de Wolf, Dieter van Melkebeek, Osamu Watanabe, Henry Yuen, Scott Aaronson, William Kretschmer and Eric Allender for helpful comments and remarks. SG was supported by the DFG under grant numbers 432788384 and 450041824, the BMBF within the funding program “Quantum Technologies - from Basic Research to Market” via project PhoQuant (grant number 13N16103), and the project “PhoQC” from the programme “Profilbildung 2020”, an initiative of the Ministry of Culture and Science of the State of Northrhine Westphalia. JK was supported by the DFG under grant number 450041824.

\bibliographystyle{alpha}
\bibliography{biblio.bib}

\newcommand{\etalchar}[1]{$^{#1}$}
\begin{thebibliography}{DKMW13}

\bibitem[AA11]{aaronson2011computational}
Scott Aaronson and Alex Arkhipov.
\newblock The computational complexity of linear optics.
\newblock In {\em Proceedings of the forty-third annual ACM symposium on Theory
  of computing}, pages 333--342, 2011.

\bibitem[Aar14]{aaronson2014equivalence}
Scott Aaronson.
\newblock The equivalence of sampling and searching.
\newblock {\em Theory of Computing Systems}, 55(2):281--298, 2014.

\bibitem[AB09]{arora2009computational}
Sanjeev Arora and Boaz Barak.
\newblock {\em Computational complexity: a modern approach}.
\newblock Cambridge University Press, 2009.

\bibitem[ABK23]{aaronson2023qubit}
Scott Aaronson, Harry Buhrman, and William Kretschmer.
\newblock A qubit, a coin, and an advice string walk into a relational problem.
\newblock {\em arXiv preprint arXiv:2302.10332}, 2023.

\bibitem[AIK22]{aaronson2021acrobatics}
Scott Aaronson, DeVon Ingram, and William Kretschmer.
\newblock {The Acrobatics of BQP}.
\newblock In Shachar Lovett, editor, {\em 37th Computational Complexity
  Conference (CCC 2022)}, volume 234 of {\em Leibniz International Proceedings
  in Informatics (LIPIcs)}, pages 20:1--20:17, Dagstuhl, Germany, 2022. Schloss
  Dagstuhl -- Leibniz-Zentrum f{\"u}r Informatik.

\bibitem[Amb99]{ambainis1999better}
Andris Ambainis.
\newblock A better lower bound for quantum algorithms searching an ordered
  list.
\newblock In {\em 40th Annual Symposium on Foundations of Computer Science
  (Cat. No. 99CB37039)}, pages 352--357. IEEE, 1999.

\bibitem[Amb00]{ambainis2000quantum}
Andris Ambainis.
\newblock Quantum lower bounds by quantum arguments.
\newblock In {\em Proceedings of the thirty-second annual ACM symposium on
  Theory of computing}, pages 636--643, 2000.

\bibitem[BBBV97]{bennett1997strengths}
Charles~H Bennett, Ethan Bernstein, Gilles Brassard, and Umesh Vazirani.
\newblock Strengths and weaknesses of quantum computing.
\newblock {\em SIAM journal on Computing}, 26(5):1510--1523, 1997.

\bibitem[BBC{\etalchar{+}}01]{beals2001quantum}
Robert Beals, Harry Buhrman, Richard Cleve, Michele Mosca, and Ronald {de
  Wolf}.
\newblock Quantum lower bounds by polynomials.
\newblock {\em Journal of the ACM (JACM)}, 48(4):778--797, 2001.

\bibitem[BDCG89]{ben1989theory}
Shai Ben-David, Benny Chor, and Oded Goldreich.
\newblock On the theory of average case complexity.
\newblock In {\em Proceedings of the twenty-first annual ACM symposium on
  Theory of computing}, pages 204--216, 1989.

\bibitem[BV93]{bernstein1993quantum}
Ethan Bernstein and Umesh Vazirani.
\newblock Quantum complexity theory.
\newblock In {\em Proceedings of the twenty-fifth annual ACM symposium on
  Theory of computing}, pages 11--20, 1993.

\bibitem[CMV16]{chakraborty2016algorithmic}
Supratik Chakraborty, Kuldeep~S Meel, and Moshe~Y Vardi.
\newblock Algorithmic improvements in approximate counting for probabilistic
  inference: from linear to logarithmic {SAT} calls.
\newblock In {\em Proceedings of the Twenty-Fifth International Joint
  Conference on Artificial Intelligence}, pages 3569--3576, 2016.

\bibitem[DKMW13]{dell2013valiant}
Holger Dell, Valentine Kabanets, {Dieter van} Melkebeek, and Osamu Watanabe.
\newblock Is {V}aliant--{V}azirani’s isolation probability improvable?
\newblock {\em computational complexity}, 22:345--383, 2013.

\bibitem[DM22]{delannoy2022almost}
Remi Delannoy and Kuldeep~S Meel.
\newblock On almost-uniform generation of {SAT} solutions: The power of 3-wise
  independent hashing.
\newblock In {\em Proceedings of the 37th Annual ACM/IEEE Symposium on Logic in
  Computer Science}, pages 1--10, 2022.

\bibitem[Gro96]{grover1996fast}
Lov~K Grover.
\newblock A fast quantum mechanical algorithm for database search.
\newblock In {\em Proceedings of the twenty-eighth annual ACM symposium on
  Theory of computing}, pages 212--219, 1996.

\bibitem[INN{\etalchar{+}}22]{irani2021quantum}
Sandy Irani, Anand Natarajan, Chinmay Nirkhe, Sujit Rao, and Henry Yuen.
\newblock {Quantum Search-To-Decision Reductions and the State Synthesis
  Problem}.
\newblock In Shachar Lovett, editor, {\em 37th Computational Complexity
  Conference (CCC 2022)}, volume 234 of {\em Leibniz International Proceedings
  in Informatics (LIPIcs)}, pages 5:1--5:19, Dagstuhl, Germany, 2022. Schloss
  Dagstuhl -- Leibniz-Zentrum f{\"u}r Informatik.

\bibitem[IP01]{impagliazzo2001complexity}
Russell Impagliazzo and Ramamohan Paturi.
\newblock On the complexity of k-{SAT}.
\newblock {\em Journal of Computer and System Sciences}, 62(2):367--375, 2001.

\bibitem[Kre86]{krentel1986complexity}
Mark~W Krentel.
\newblock The complexity of optimization problems.
\newblock In {\em Proceedings of the eighteenth annual ACM symposium on Theory
  of computing}, pages 69--76, 1986.

\bibitem[KRW12]{kawachi2012query}
Akinori Kawachi, Benjamin Rossman, and Osamu Watanabe.
\newblock Query complexity and error tolerance of witness finding algorithms.
\newblock In {\em Electron. Colloquium Comput. Complex.}, volume~19, page~2,
  2012.

\bibitem[KRW17]{kawachi2017query}
Akinori Kawachi, Benjamin Rossman, and Osamu Watanabe.
\newblock The query complexity of witness finding.
\newblock {\em Theory of Computing Systems}, 61:305--321, 2017.

\bibitem[Sto83]{stockmeyer1983complexity}
Larry Stockmeyer.
\newblock The complexity of approximate counting.
\newblock In {\em Proceedings of the fifteenth annual ACM symposium on Theory
  of computing}, pages 118--126, 1983.

\bibitem[VV85]{valiant1985np}
Leslie~G Valiant and Vijay~V Vazirani.
\newblock {NP} is as easy as detecting unique solutions.
\newblock In {\em Proceedings of the seventeenth annual ACM symposium on Theory
  of computing}, pages 458--463, 1985.

\end{thebibliography}

\end{document}